%% file: main.tex
\documentclass[submission]{eptcs}
\providecommand{\event}{} 

\input{head_pc}
\RequirePackage[normalem]{ulem}
\RequirePackage{color}\definecolor{RED}{rgb}{1,0,0}\definecolor{BLUE}{rgb}{0,0,1}

\newtoggle{appendixing} 
\togglefalse{appendixing}

\title{Higher-order Processes with Parameterization over Names and Processes}
\author{Xian Xu
\thanks{This work has been supported by project ANR 12IS02001 PACE and NSF of China (61261130589,
61472239, 61572318).}
\institute{East China University of Science and Technology, China}
\email{xuxian@ecust.edu.cn}
}
\def\titlerunning{Higher-order Processes with Parameterization}
\def\authorrunning{Xian Xu}
\begin{document}
\maketitle

\begin{abstract}
Parameterization extends higher-order processes with the capability of abstraction and application (like those in lambda-calculus). This extension is strict, i.e., higher-order processes equipped with parameterization is computationally more powerful. This paper studies higher-order processes with two kinds of parameterization: one on names and the other on processes themselves. We present two results. One is that in presence of parameterization, higher-order processes can encode first-order (name-passing) processes in a quite neat fashion, in contrast to the fact that higher-order processes without parameterization cannot encode first-order processes at all. In the other result, we provide a simpler characterization of the (standard) context bisimulation for higher-order processes with parameterization, in terms of the normal bisimulation that stems from the well-known normal characterization for higher-order calculus. These two results demonstrate more essence of the parameterization method in the higher-order paradigm toward expressiveness and behavioural equivalence.

\vspace*{.1cm}
\noindent\emph{keywords}: Parameterization, Context bisimulation, Higher-order, First-order, Processes 
\end{abstract}

\input{introduction.tex}
\input{preliminary.tex}
\input{encoding.tex}
\input{normal.tex}
\input{conclusion.tex}
\sepp
\noindent\textbf{Acknowledgements}\;\;
We thank the anonymous referees for their useful comments on this article. 

\bibliographystyle{eptcs}
\bibliography{process}

\iftoggle{appendixing}{%
\clearpage
\appendix
\noindent\textbf{\Large Appendix}
\input{appendix_proof_encoding.tex}
\input{appendix_proof_normal.tex}
}{%
}


\end{document}

%% file: head_pc.tex
\usepackage{etex}
\usepackage[usenames]{color}
\usepackage{makeidx}
\usepackage{verbatim}
\usepackage{fancyhdr}
\usepackage{fancybox}
\usepackage{stmaryrd}
\usepackage[reqno]{amsmath}
\usepackage{amsthm}
\usepackage{amssymb}
\usepackage{amsfonts}
\usepackage{bm}
\usepackage{amsbsy}
\usepackage[all]{xypic}
\usepackage {indentfirst}
\usepackage{graphicx}
\usepackage[hang]{subfigure}
\usepackage{cite}
\usepackage{proof}
\usepackage{bbding}
\usepackage{pgf}
\usepackage{tikz}
\usetikzlibrary{positioning,arrows,automata}
\usepackage{cancel}
\usepackage{semantic}
\usepackage{extarrows}
\usepackage{etoolbox} 

 \newtheorem{theorem}{Theorem}
 
 \newtheorem{lemma}[theorem]{Lemma}

 \theoremstyle{definition}
 \newtheorem{definition}[theorem]{Definition}

 \theoremstyle{remark}

\newcommand{\HOPi}{\ensuremath{\Pi}}

\newcommand{\FOPi}{\ensuremath{\pi}}

\newcommand{\HOPid}{\ensuremath{\Pi^d}}
\newcommand{\HOPiD}{\ensuremath{\Pi^D}}
\newcommand{\HOPiDd}{\ensuremath{\Pi^{D,d}}}
\newcommand{\para}{\,|\,}
\newcommand{\fn}[1]{\mbox{\rm fn($#1$)}} 
\newcommand{\bn}[1]{\mbox{\rm bn($#1$)}} 
\newcommand{\n}[1]{\mbox{\rm n($#1$)}} 
\newcommand{\fnv}[1]{\mbox{\rm fnv($#1$)}} 
\newcommand{\bnv}[1]{\mbox{\rm bnv($#1$)}} 
\newcommand{\nv}[1]{\mbox{\rm nv($#1$)}} 
\newcommand{\fpv}[1]{\mbox{\rm fpv($#1$)}} 
\newcommand{\bpv}[1]{\mbox{\rm bpv($#1$)}} 
\newcommand{\pv}[1]{\mbox{\rm pv($#1$)}} 

\newcommand{\encode}[3]{ \lds #1 \rds^{#2}_{#3}}

\newcommand{\fosub}[2]{\{#1/#2\} }
\newcommand{\hosub}[2]{\{#1/#2\} }
\newcommand{\vect}[1]{{#1_1,#1_2,...,#1_n} }
\newcommand{\ve}[1]{\widetilde{#1}}

\newcommand{\size}[1]{|#1|}
\newcommand{\st}[1]{\,{\xrightarrow{#1}}\, }

\newcommand{\wt}[1]{{\xLongrightarrow{#1}} }
\newcommand{\rc}[1]{{\color{red} #1}}
\newcommand{\rrc}[1]{{#1}} 
\newcommand{\bc}[1]{{\color{blue} #1}}
\newcommand{\rbc}[1]{{#1}} 
\newcommand{\sep}{\vspace*{0.7cm}}
\newcommand{\sepp}{\vspace*{0.3cm}}

\newcommand{\nsepv}[1]{\vspace{0mm}}
\newcommand{\nsepvs}[1]{\vspace{-#1cm}} 
\newcommand{\psepvs}[1]{\vspace{#1cm}} 

\newcommand{\SE}{\equiv }

\newcommand{\lds}{[\![}
\newcommand{\rds}{]\!]}
\newcommand{\encoding}[3]{ \lds #1 \rds^{#2}_{#3}}

\newcommand{\enc}[1]{\llbracket #1 \rrbracket}
\newcommand{\DEF}{\stackrel{\textrm{def}}{=}} 
\newcommand{\myxcancel}[1]{\xcancel{#1}} 
\newcommand{\lrangle}[1]{\langle #1 \rangle} 
 
\newcommand{\EWLB}{\ensuremath{\approx_{l}} } 
\newcommand{\SGB}{\ensuremath{\sim_{g}} }
\newcommand{\WGB}{\ensuremath{\approx_{g}} }
\newcommand{\WCB}{\ensuremath{\approx_{ct}} }
\newcommand{\WWCB}{\ensuremath{\dot\approx_{ct}} }
\newcommand{\SCB}{\ensuremath{\sim_{ct}} }
\newcommand{\WSCB}{\ensuremath{\dot\sim_{ct}} }
\newcommand{\WNB}{\ensuremath{\approx_{nr}} }
\newcommand{\SNB}{\ensuremath{\sim_{nr}} }
\newcommand{\diverge}[1]{#1^{\Uparrow}}
\newcommand{\trigger}{\ensuremath{Tr_m} } 
\newcommand{\triggerD}{\ensuremath{Tr_m^D} } 
\newcommand{\triggerd}{\ensuremath{Tr_m^d} } 

\makeatletter
\@ifundefined{newslide}{%
 
}{%
}
\makeatletter
\@ifundefined{endinput}{%
 
}{%
}

\newcommand{\stress}[1]{\rc{#1}} 

\newcommand{\tdup}[1]{}



%% file: introduction.tex
\section{Introduction}\label{s:introduction}

In concurrent systems, higher-order means that processes communicate by means of process-passing (i.e., program-passing), whereas first-order means that processes communicate through name-passing (i.e., reference-passing). Parameterization originates from lambda-calculus (which is itself of higher-order nature), and enables processes, in a concurrent setting, to do abstraction and application in a way similar to that of lambda-calculus. Say $P$ is a higher-order process, then an abstraction $\lrangle{U}P$ means abstracting the variable $U$ in $P$ to obtain somewhat a function (like $\lambda U.P$ in terms of lambda-calculus), and correspondingly an application $(\lrangle{U}P)\lrangle{K}$ means applying process $K$ to the abstraction and obtaining an instantiation $P\hosub{K}{U}$ (i.e., replacing each variable $U$ in $P$ with $K$, like $(\lambda U.P)K$ in terms of lambda-calculus). There are basically two kinds of parameterization: parameterization on names and parameterization on processes. In the former, $U$ is a name variable and $K$ is a concrete name. In the latter, $U$ is a process variable and $K$ is a concrete process.
Parameterization is a natural way to extend the capacity of higher-order processes and this extension is strict, that is, the computational power strictly increases with the help of parameterization \cite{LPSS10}. In this paper, we study higher-order processes in presence of parameterization.


Comparison between higher-order and first-order processes is a frequent topic in concurrency theory. Such comparison, for example, asks whether higher-order processes can correctly express first-order processes, or vice versa. It is well known that first-order processes can elegantly encode higher-order processes \cite{San92,SW01a}; the converse is however not quite the case. As the first issue, this paper addresses how to encode first-order processes with higher-order processes (equipped with parameterization).

The very early work on using higher-order process to interpret first-order ones is contributed by Thomsen \cite{Tho93}, who proposed a prototype encoding of first-order processes with higher-order processes with the relabelling operator (like that in CCS \cite{Mil89}). This encoding uses a gadget called wire to mimic the function of a name in the higher-order setting, and essentially employs the relabelling to make the wires work properly so as to fulfill the role of names. Due to the arbitrary ability of changing names (e.g., from global to local), the encoding has a correct operational correspondence (i.e., the correspondence between the processes before and after the encoding), but is very hard to analyze for full abstraction (i.e., the first-order processes are equivalent if and only if their encodings are; the `if only' direction is called soundness and the other direction is called completeness). Unfortunately, without the relabelling operator, the basic higher-order process (which has the elementary operators including input, output, parallel composition and restriction) is not capable of encoding first-order processes \cite{Xu12}. In the literature, several variants of higher-order processes are exploited to encode first-order processes. In \cite{SW01a}, an asynchronous higher-order calculus with parameterization on names is used to compile the asynchronous localized $\pi$-calculus (a variant of the first-order $\pi$-calculus \cite{MPW92}). This encoding depends heavily on the notions of `localized' which means only the output capability of a name can be communicated during interactions, and `asynchronous' which means the output is non-blocking. Though technically a nice reference, intuitively because this variant of $\pi$-calculus is less expressive than the full $\pi$-calculus, it is not very surprising that the higher-order processes with parameterization on names can interpret it faithfully, i.e., fully abstract with respect to barbed congruence. Then in \cite{XYL15}, we explore the encoding of the full $\pi$-calculus using higher-order processes with parameterization on names. In that effort, we construct an encoding that harnesses the idea of Thomsen's encoding and show that it is complete. In \cite{BHG06}, Bundgaard et al. use the HOMER to translate the name-passing $\pi$-calculus. This translation is possible because a HOMER process can, in a way quite different from parameterization, operates names in the continuation processes (resources), and this allows flexibility so that names can be communicated in an intermediate fashion. 
In \cite{KPY16}, Kouzapas et al. propose fully abstract encodings concerning first-order processes and session typed higher-order processes. Their encodings use session types to govern communications and show that in the context of session types, first-order and higher-order processes are equally expressive. This work is well related to those mentioned above (and that in this paper), though the context is quite different (i.e., session typed processes).

Despite the extensive research on encoding first-order processes with (variant) higher-order processes, the following question has remained open: \emph{Is there an encoding of first-order processes by the higher-order processes with the capability of parameterization?}
This question is important in two aspects. One is that parameterization brings about the core of lambda-calculus to higher-order concurrency, so it appears reasonable for such an extension to be able to express first-order processes which has long been shown to be capable of expressing the lambda-calculus. Knowing how this can be achieved would be interesting. The other is that the converse has a almost standard encoding method, i.e., encoding variants of higher-order processes with first-order processes. Yet higher-order processes are still short of an effective way to express first-order ones. Resolving this can also provide (technical) reference for practical work beyond the encoding itself.

%
%


Closely related with the first issue on expressiveness, the second issue this paper deals with is the characterization of bisimulation on higher-order processes. Bisimulation theory is a pivotal part of a process model, including the higher-order models, concerning which the almost standard behavioral equivalence is the context bisimulation \cite{San92}. The central idea of context bisimulation is that when comparing output actions, the transmitted process and the residual process (i.e., the process obtained after sending a process) are considered at the same time, rather than separately (like in the applicative higher-order bisimulation proposed by Thomsen \cite{Tho90, Tho93}). For example (for simplicity we do not consider local names), if $P$ and $Q$ are context bisimilar and $P\st{\overline{a}A} P'$ (i.e., $P$ outputs $A$ on $a$ and becomes $P'$), then $Q\wt{\overline{a}B} Q'$ (i.e., $Q$ outputs $B$ on $a$ possibly involving some internal actions and becomes $Q'$), and for every (receiving) environment $E[\cdot]$, $P'\para E[A]$ and $Q'\para E[B]$ are still context bisimilar (here $\para$ denotes concurrency, and $E[A]$ means putting $A$ in the environment $E$). However, in its original form, context bisimulation suffers from inconvenience to use, because it calls for checking with regard to every possible receiving environment. This leads to works on the simpler characterization, called normal bisimulation, of the context bisimulation. The central idea of normal bisimulation, proposed by Sangiorgi \cite{San92, SW01a}, is that instead of checking with a general process in input and a general context in output, one only needs to comply with the matching of some special process or context, specifically a class of terms called triggers. To meet this challenge, a crucial so-called factorization theorem is used to circumvent technical difficulty. We briefly explain how normal bisimulation is designed in the basic higher-order processes. In particular, the factorization states the following property, where $\WCB$ denotes context bisimulation, and $\overline{m}.P$ and $m.P$ are CCS-like prefixes in which the communicated contents are not important \cite{SW01a}. \nsepvs{.3}
\[ E[A] \WCB (m)(E[\overline{m}.0] \para !m.A)  \nsepvs{.3} 
\] 
One can clearly identify the reposition of the process $A$ of interest, which in fact captures the core of the property: move $A$ to a new position as a repository, which in turn can be retrieved as many times as needed in the original environment $E$, with the help of the pointer undertaken by the fresh channel $m$ (called trigger). Inspired by the factorization, normal bisimulation can be developed. We take the output as an example (input is similar), and restriction operation in output is omitted for the sake of simplicity. As stated above, context bisimulation requires the following chasing diagram, which is now extended with an application of the factorization. \nsepvs{.3} 
\[ \nsepvs{.2} 
\xymatrix@C=20pt{
 &  & P \ar@{.}[rr]|-{\WCB}\ar@{->}[d]_{\overline{a}A}  &  & Q \ar@{=>}[d]^{\overline{a}B}  &  & \\
P'\para (m)(E[\overline{m}.0] \para !m.A) \ar@{}[r]|-{\WCB} & P'\para E[A] \ar@/_1.6pc/@{.}[0,4]|{\WCB} & P'  & &  Q'  & Q'\para E[B] \ar@{}[r]|-{\WCB}  &  Q'\para (m)(E[\overline{m}.0] \para !m.B)
}
\]
Since context bisimulation $\WCB$ is a congruence, one can cancel the common part of (the leftmost) $P'\para (m)(E[\overline{m}.0] \para !m.A)$ and (the rightmost) $Q'\para (m)(E[\overline{m}.0] \para !m.B) $, and simply requires that $P'\para !m.A$ and $Q'\para !m.B$ are related, without fearing losing any discriminating power. This in turn leads to the following requirement in normal bisimulation (assuming $\mathcal{R}$ is a normal bisimulation). \nsepvs{.3} 
\[\nsepvs{.2} 
\xymatrix{
 & P \ar@{.}[rr]|-{\mathcal{R}}\ar@{->}[d]_{\overline{a}A}  &  & Q \ar@{=>}[d]^{\overline{a}B}  &   \\
 P'\para !m.A \ar@/_1.6pc/@{.}[0,4]|{\mathcal{R}} & P'  & &  Q'  & Q'\para !m.B
}
\]

Subsequent works attempt to extend the normal bisimulation to variants of higher-order processes. In Sangiorgi's initial work \cite{San92}, normal bisimulation is also obtained for higher-order processes with parameterization. That characterization , however, is made in the presence of first-order processes (i.e., name-passing), and thus not very convincing with regard to the inner complexity of context bisimulation in presence of parameterization. In \cite{Xu13}, we revisited this issue and show that in a purely higher-order setting (viz., no name-passing at all), parameterization on processes does not deprive one of the convenience of normal bisimulation. Although the idea is inspired by the original work of Sangiorgi, the proof approach is more direct. In \cite{LSS09, LSS11}, Lenglet et al. study higher-order processes with passivation (i.e., the process in the output position may evolve), and report a normal bisimulation for a sub-calculus without the restriction operator, but that characterization has somewhat a different flavor, since the higher-order bisimulation \cite{Tho93} rather than the context bisimulation is taken. Though these works carry out insightful research and give meaningful references, it is currently still not clear how to construct a simple characterization of context bisimulation based on parameterization over names, and this raises the following fundamental question: \emph{Does higher-order processes with parameterization on names have a normal bisimulation?}
In the second part of this paper, we move further from \cite{Xu13,San92}, and offer a normal bisimulation for higher-order processes in the setting of parameterization over both names and processes.  



%

\psepvs{.2}
\noindent\textbf{Contribution}~ 
In summary, our contribution of this work is as follows.
\begin{itemize}
\item 
We show that the extension with parameterization (on both names and processes) allows higher-order processes to interpret first-order processes in a surprisingly concise yet elegant manner. Such kind of encoding is of a somewhat dissimilar flavor, and moreover not possible in absence of parameterization. We give the detailed encoding strategy, and prove that it satisfies a number of desired properties well-known in the field. \\
The idea of the encoding in this paper is quite different from our abovementioned work in \cite{XYL15}, where we build an encoding that allows parameterization merely on names (i.e., no parameterization on processes). The soundness of that encoding is not very satisfying, which in a sense defeats some purpose of the encoding, and this actually precipitates the work here.

\item 
We establish the normal bisimulation, as an effectively simpler characterization of context bisimulation, for higher-order processes with both kinds of parameterization. This normal bisimulation extends those for higher-order processes without parameterization, particularly in the manipulation of abstractions on names. As far as we are concerned, similar characterization has not been reported before.\\
That the processes are purely higher-order (that is, without name-passing) improves the result in \cite{San92}, and articulates that the characterization based on normal bisimulation is a property independent of first-order name-passing. Moreover, this does not contradict the argument in \cite{Xu13} that there is little hope that normal bisimulation exists in higher-order processes with (only) parameterization on names, because here the processes are capable of parameterization on processes as well (though still higher-order).
\end{itemize}

\noindent\textbf{Organization}
The remainder of this paper is organized as below. In Section \ref{s:preliminary}, we introduce the calculi and a notion of encoding used in this paper. In Section \ref{s:encoding}, we present the encoding from first-order processes to higher-order processes with parameterization, and discuss its properties. In Section \ref{s:normal}, we define the normal bisimulation for higher-order processes with parameterization, and prove that it truly characterizes context bisimulation. Section \ref{s:conclusion} concludes this work and point out some further directions.


%% file: preliminary.tex
\section{Preliminary}\label{s:preliminary}
In this section we present the basic definitions and notations used in this work. \psepvs{.2}

\noindent\textbf{\large 2.1~ Calculus \FOPi}


The first-order (name-passing) pi-calculus, \FOPi, is proposed by Milner et al. \cite{MPW92}. For the sake of simplicity, throughout the paper, names (ranged over by $m,n,u,v,w$) are divided into two classes: name constants (ranged over by $a,b,c,d,e$) and name variables (ranged over by $x,y,z$) \cite{EN86a,EN00,Fu15}. The grammar is as below with the constructs having their standard meaning. We note that guarded input replication is used instead of general replication, and this does not decrease the  expressiveness \cite{San98}\cite{FL09a}. \nsepvs{.2}
\[P,Q := 0 \,\Big{|}\, m(x).P \,\Big{|}\, \overline{m}n.P \,\Big{|}\,  (c)P \,\Big{|}\, P\para Q \,\Big{|}\, !m(x).P 
\]

A name constant $a$ is bound (or local) in $(a)P$ and free (or global) otherwise. A name variable $x$ is bound in $a(x).P$ and free otherwise.
Respectively $\fn{\cdot}, \bn{\cdot}, \n{\cdot}, \fnv{\cdot}, \bnv{\cdot}, \nv{\cdot}$ denote free name constants, bound name constants, names, free name variables, bound name variables, and name variables in a set of processes. A name is fresh if it does not appear in any process under discussion. By default, closed processes are considered, i.e., those having no free variables.
As usual, here are a few derived operators: $\overline{a}(d).P \DEF (d)\overline{a}d.P$, $a.P \DEF  a(x).P \;(x\notin fv(P))$, $\overline{a}.P \DEF \overline{a}(d).P\;(d\notin fn(P))$; $\tau.P\DEF (a)(a.P\para \overline{a}.0)$ ($a$ fresh). A trailing $0$ process is usually omitted.
We denote tuples by a tilde. For tuple $\ve{n}$: $\size{\ve{n}}$ denotes its length; 
$m\ve{n}$ denotes incorporating $m$. Multiple restriction $(c_1)(c_2)\cdots (c_k)E$ is abbreviated as $(\ve{c})E$.
Substitution $P\fosub{n}{m}$ is a mapping that replaces $m$ with $n$ in $P$ while keeping the rest unchanged. 
A context $C$ is a process with some subprocess replaced by the hole $[\cdot]$, and $C[A]$ is the process obtained by filling in the hole by $A$.

The semantics of \FOPi\ is defined by the LTS (Labelled Transition System) below. \nsepvs{.2} 
\[
\begin{array}{lllll}
\infer{a(x).P\st{a(b)} P\fosub{b}{x}}{} \quad &
\infer{\overline{a}b.P\st{\overline{a}b} P}{} \quad &
\infer{!a(x).P \st{a(b)} P\fosub{b}{x}\para !a(x).P}{}\quad  & 
\infer[{\scriptstyle c\not\in n(\lambda)}]{(c)P\st{\lambda} P'}{P\st{\lambda} P'} \quad
\end{array}
\]
\[
\begin{array}{llll}
\infer[{\scriptstyle c\neq a}]{(c)P\st{\overline{a}(c)} P'}{P\st{\overline{a}c} P'} \quad &
\infer[{\scriptstyle bn(\lambda)\cap fn(Q)=\emptyset}]{P\para Q\st{\lambda} P'\para Q}{P\st{\lambda} P'} \quad &
\infer{P\para Q\st{\tau}P'\para Q'}{P\st{a(b)}P' \quad Q\st{\overline{a}b} Q'}\quad &
\infer{P\para Q\st{\tau}(b)(P'\para Q')}{P\st{a(b)} P'\quad Q\st{\overline{a}(b)} Q'} 
\end{array}
\]
Actions, ranged over by $\lambda,\alpha$, comprise internal move $\tau$, and visible ones: input ($a(b)$), output ($\overline{a}b$) and bound output ($\overline{a}(c)$). We note that actions occur only on name constants, and a communicated name is also a constant. 
We denote by $\equiv$ the standard structural congruence \cite{MPW92}\cite{SW01a}, which is the smallest relation satisfying the monoid laws for parallel composition, commutative laws for both composition and restriction, and a distributive law $(c)(P\para Q) \SE (c)P\para Q$ (if $c\notin fn(Q)$).
We use $\wt{}$ for the reflexive transitive closure of $\st{\tau}$, $\wt{\lambda}$ for $\wt{}\st{\lambda}\wt{}$, and $\wt{\widehat{\lambda}}$ for $\wt{\lambda}$ if $\lambda$ is not $\tau$, and $\wt{}$ otherwise. A process $P$ is divergent, denoted $\diverge{P}$, if it has an infinite sequence of $\tau$ actions.

Throughout the paper, we use the following standard notion of ground bisimulation\cite{MPW92,EN00,SW01a}.  
\begin{definition}
A ground bisimulation is a symmetric relation $\mathcal{R}$ on \FOPi\ processes s.t. whenever $P\,\mathcal{R}\, Q$ the following property holds:
If $P\st{\alpha} P'$ where $\alpha$ is $a(b)$, $\overline{a}b$, $\overline{a}(b)$, or $\tau$, then $Q\wt{\hat{\alpha}} Q'$ for some $Q'$ and $P'\,\mathcal{R}\, Q'$.
Ground bisimilarity, $\WGB$, is the largest ground bisimulation.
\end{definition}
We denote by $\SGB$ the strong ground bisimilarity (i.e., replacing $\wt{\widehat{\alpha}}$ with  $\st{\alpha}$ in the definition). It is well-known that $\WGB$ is a congruence \cite{EN00,SW01a}, and coincides with the so-called local bisimilarity as defined below \cite{Fu05b,Xu12}.
\begin{definition}\label{ext-local-bisi}
Local bisimilarity $\EWLB$ is the largest symmetric local bisimulation relation $\mathcal{R}$ on $\pi$ processes such that:
(1) if $P\st{\lambda} P'$, $\lambda$ is not bound output, then $Q\wt{\widehat{\lambda}} Q'$ and $P'\,\mathcal{R}\, Q'$;
(2) if $P\st{\overline{a}(b)} P'$, then $Q\wt{\overline{a}(b)} Q'$, and for every $R$, $(b)(P'\para R)\,\mathcal{R}\, (b)(Q'\para R)$.
\end{definition}


\psepvs{.2}
\noindent\textbf{\large 2.2 ~Calculus \HOPiDd}

For the sake of conciseness, we first define the basic higher-order calculus and then the extension with parameterizations.

\psepvs{.2}
\noindent\textbf{2.2.1 ~ Calculus \HOPi}

The basic higher-order (process-passing) calculus, \HOPi, is defined by the following grammar in which the operators have their standard meaning.
We denote by $X,Y,Z$ process variables.
\[
\begin{array}{l}
T,T' ::= 0 \,\Big{|}\, X \,\Big{|}\,  u(X).T \,\Big{|}\, \overline{u}T'.T \,\Big{|}\, T\para T' \,\Big{|}\, (c)T \,\Big{|}\,  !u(X).T \,\Big{|}\, !\overline{u}T'.T
\end{array}
\]
We use $a.0$ for $a(X).0$, $\overline{a}.0$ for $\overline{a}0.0$, 
$\tau.P$ for $(a)(a.P\para \overline{a}.0)$, and sometimes $\overline{a}[A].T$ for $\overline{a}A.T$. Like \FOPi, a tilde represents a tuple.
We reuse the notations for names in \FOPi\, and additionally use $\fpv{\cdot}$, $\bpv{\cdot}$, $\pv{\cdot}$ respectively to denote free process variables, bound process variables and process variables in a set of processes. Closed processes are those having no free variables.
A higher-order substitution $T\hosub{A}{X}$ replaces variable $X$ with $A$ and can be extended to tuples in the usual way.
$E[\ve{X}]$ denotes $E$ with (possibly) variables $\ve{X}$, and $E[\ve{A}]$ stands for $E\hosub{\ve{A}}{\ve{X}}$.
The guarded replications used in the grammar can actually be derived \cite{Tho93,LPSS08}, and we make them primitive for convenience. 
The semantics of \HOPi\ is as below. \nsepvs{.2}
\[
\begin{array}{lllll}
\frac{\displaystyle }{\displaystyle a(X).T\st{a(A)} T\hosub{A}{X}}  &
 \frac{}{\displaystyle \overline{a}A.T\st{\overline{a}A} T}  &
\frac{\displaystyle T\st{\lambda} T'}{\displaystyle (c)T\st{\lambda} (c)T'}{\scriptstyle c\not\in n(\lambda)} &
\frac{\displaystyle T\st{\lambda} T'}{\displaystyle T\para T_1\st{\lambda} T'\para T_1} & 
 \frac{}{\displaystyle !\overline{a}A.T\st{\overline{a}A} T \para !\overline{a}A.T}
\end{array}
\]
\[\begin{array}{lll}
\frac{\displaystyle T\st{(\ve{c})\overline{a}[A]} T'}{\displaystyle (d)T\st{(d)(\ve{c})\overline{a}[A]} T'} {\scriptstyle d \in fn(A){-}\{\ve{c},a\}}  &
 \frac{\displaystyle T_1\st{a(A)} T_1', T_2\st{(\ve{c})\overline{a}[A]} T_2'}{\displaystyle T_1\para T_2 \st{\tau}(\ve{c})(T_1'\,|\,T_2')} & 
\frac{\displaystyle }{\displaystyle !a(X).T\st{a(A)} T\hosub{A}{X}\para !a(X).T }
\end{array}
\]
We denote by $\alpha,\lambda$ the actions: internal move ($\tau$), input ($a(A)$), output ($(\ve{c})\overline{a}A$) in which $\ve{c}$ is some local names carried by $A$ during the output. We always assume no name capture with resort to $\alpha$-conversion.
The notations $\wt{}$, $\wt{\lambda}$ and $\wt{\widehat{\lambda}}$ are similar to those in \FOPi. We also reuse $\equiv$ for the structural congruence in \HOPi\ (and also \HOPiDd\ to be defined shortly) \cite{SW01a}, and this shall not raise confusion under specific context.

\psepvs{.2}
\noindent\textbf{2.2.2 ~ Calculus \HOPiDd}

Parameterization extends \HOPi\ with the syntax and semantics below. 
Symbol $U_i$ (respectively, $K_i$) ($i=1,...,n$) is used as a meta-parameter of an abstraction (respectively, meta-instance of an application), and stands for a process variable or name variable (respectively, a process or a name). \nsepvs{.2} 
\[
\begin{array}{ll}
\mbox{\small Extension of syntax: } & \lrangle{\vect{U}} T \;\Big{|}\;  T'\lrangle{\vect{K}}\\
\mbox{\small Extension of semantics: } & \frac{\displaystyle Q\equiv P\quad P\st{\lambda} P'\quad P'\equiv Q'}{\displaystyle Q\st{\lambda} Q'} \\
\mbox{\small Extension of structural congruence ($\equiv$): } & F\lrangle{\ve{K}} \equiv T\hosub{\ve{K}}{\ve{U}} \quad \mbox{ where } F\DEF \lrangle{\ve{U}}T \mbox{ and } \size{\ve{U}}{=}\size{\ve{K}} 
\end{array}
\]

We denote by $\lrangle{\vect{U}} T$ an n-ary abstraction in which $\vect{U}$ are the parameters to be instantiated during the application $T'\lrangle{\vect{K}}$ in which the parameters are replaced by instances $\vect{K}$. This application is modelled by an extensional rule for structural congruence as above, in combination with the usual LTS rule for structural congruence as well, so as to make the process engaged in application evolve effectively.
The condition $\size{\ve{U}}{=}\size{\ve{K}}$ requires that the parameters and the instantiating objects should be equal in length.

Now parameterization on process is obtained by taking $\ve{U},\ve{K}$ as $\ve{X},\ve{T'}$ respectively, and parameterization on names is obtained by taking $\ve{U},\ve{K}$ as $\ve{x},\ve{u}$ respectively. The corresponding abstractions are sometimes called process abstraction and name abstraction respectively. For convenience, names are handled in the same way as that in \FOPi\ (so are the related notations). We denote by \HOPiDd\ the calculus \HOPi\ extended with both kinds of parameterizations.
Calculus \HOPiDd\ can be made more precise with the help of a type system \cite{San92} which however is not important for this work and not presented.
We note that in $\lrangle{\vect{U}} T$, variables $\vect{U}$ are bound.

Throughout the paper, we reply on the following notion of context bisimulation \cite{San92,San94}. 
\begin{definition}\label{context-bisimulation}
A symmetric relation $\mathcal{R}$ on \HOPiDd\ processes is a context bisimulation, if $P\,\mathcal{R}\, Q$ implies the following properties: 
(1) if $P \st{\alpha} P'$ and $\alpha$ is $a(A)$ or $\tau$, then $Q \wt{\widehat{\alpha}} Q'$ for some $Q'$ and $P'\,\mathcal{R}\, Q'$; \\
(2) if $P \st{(\ve{c})\overline{a}A} P'$ and $A$ is a process abstraction or name abstraction or not an abstraction, then $Q \wt{(\ve{d})\overline{a}B} Q'$ for some $B$ that is accordingly a process abstraction or name abstraction or not an abstraction, and moreover for every $E[X]$ s.t. $\{\ve{c},\ve{d}\}\cap fn(E)=\emptyset$ it holds that $(\ve{c})(E[A]\para P') \; \mathcal{R}\;  (\ve{d})(E[B]\para Q')$.
Context bisimilarity, written $\WCB$, is the largest context bisimulation.
\end{definition}
We note that the matching for output in context bisimulation is required to bear the same kind of communicated process as compared to the simulated action. Relation $\SCB$ denotes the strong context bisimilarity. As is well-known, $\WCB$ is a congruence \cite{San92,San94}.

\psepvs{.2}
\noindent\textbf{\large 2.3 ~A notion of encoding}

We define a notion of encoding in this section. 
We assume a process model $\mathcal{L}$ is a triplet $(\mathcal{P}, \st{}, \approx)$, where $\mathcal{P}$ is the set of processes, $\st{}$ is the  LTS with a set $\mathcal{A}$ of actions, and $\approx$ is a behavioral equivalence.
Given $\mathcal{L}_i\DEF (\mathcal{P}_i, \st{}_i, \approx_i)$ ($i{=}1,2$), an encoding from  $\mathcal{L}_1$ to $\mathcal{L}_2$ is a function $\encoding{\cdot}{}{}: \mathcal{P}_1 \longrightarrow \mathcal{P}_2$ that satisfies some set of criteria.
Notation $\encode{\mathcal{P}_1}{}{}$ stands for the image of the $\mathcal{L}_1$-processes inside $\mathcal{L}_2$ under the encoding. It should be clear that  $\encode{\mathcal{P}_1}{}{}\subseteq \mathcal{P}_2$. We use $\dot\approx_2$ to denote the behavioural equivalence $\approx_2$ restricted to $\encode{\mathcal{P}_1}{}{}$ \cite{Gor09}.
The following criteria set (Definition \ref{gorla-like-cond}) used in this paper stems from \cite{LPSS10} (the variant \cite{LPSS10} provides is based on \cite{Gor08a}). 
As is known, encodability enjoys transitivity \cite{LPSS10}.
We will show that the encoding in Section \ref{s:encoding} satisfies all the criteria in Definition \ref{gorla-like-cond} except adequacy (1a).

\begin{definition}[Criteria for encodings]\label{gorla-like-cond}
\noindent\textbf{Static criteria}: 
(1) Compositionality. \emph{For any $k$-ary operator $op$ of $\mathcal{L}_1$, and all $P_1,...,P_k\in \mathcal{P}_1$,  $\encoding{op(P_1,...,P_k)}{}{}$ $=\, C_{op}[\encoding{P_1}{}{},...,\encoding{P_k}{}{}]$ for some (multihole) context $C_{op}[\cdots]\in \mathcal{P}_2$;}

\noindent
\textbf{Dynamic criteria}: 
(1a) Adequacy. \emph{$P \approx_1 P'$ implies $\encoding{P}{}{} \approx_2 \encoding{P'}{}{}$. This is also known as \emph{soundness}. The converse is known as \emph{completeness};} 
(1b) Weak adequacy (or weak soundness). \emph{$P \approx_1 P'$ implies $\encoding{P}{}{} \dot\approx_2 \encoding{P'}{}{}$;} ~
(2) Divergence-reflecting. \emph{If $\encoding{P}{}{}$ diverges, so does $P$.}
\end{definition}

Adequacy (1a) obviously entails weak adequacy (1b), since $\approx_2$ allows more processes in the target model $\mathcal{L}_2$ (thus more variety of contexts). Yet weak adequacy is still useful because it may be too strong if one requires the encoding process to be compatible with all kinds of contexts in the target model. For instance, in order to achieve first-order interactions in a higher-order target model, it appears quite demanding to require equivalence under all kinds of input because the target higher-order model may have more powerful computation ability (so it can feed a much involved input). So sometimes using limited contexts in the target model may be sufficient to meet the goal of the encoding. 

It is worthwhile to note that the criteria are short of those for operational correspondence. Although generally soundness and completeness appear not very informative in absence of operational correspondence (and the others) \cite{GN16,Par16}, arguably we make this choice in this work out of the following consideration. The criteria for operational correspondence used in \cite{Gor08a}, though proven useful in many models, appear not quite convenient when discussing encodings into higher-order models \cite{LPSS10}, since (for example) the case of input can be hard to comply with the criteria due to the increased complexity in the environment (namely in the context of the target higher-order model). After all, here it seems more important to have the soundness and completeness properties eventually (w.r.t. the canonical bisimulation equivalences in the source and target models), likely in a different manner of operational correspondence. Notwithstanding, we will discuss the operational correspondence of the encoding in Section \ref{s:encoding}. Moreover, as will be seen, the concrete operational correspondence in there somehow strengthens the criteria of operational correspondence (and related concepts) used in \cite{Gor08a,LPSS10} (in \cite{Gor08a} the criteria are not action-labelled and thus the notion of success sensitiveness is contrived; in \cite{LPSS10} a labelled variant criteria is posited to its purpose). 
Beyond the scope of this paper, it would be intriguing to examine the possibility of formally pinning down some variant criteria of operational correspondence having vantage for higher-order (process) models.



%% file: encoding.tex
\section{Encoding \FOPi\ into \HOPiDd}\label{s:encoding}
We show that \FOPi\ can be encoded in \HOPiDd.

\subsection{The encoding}
We have the encoding defined as below (being homomorphism on the other operators, except that the encoding of input guarded replication is defined as $\enc{!m(x).P}\DEF !\enc{m(x).P}$).
\[
\begin{array} {lrcll}
 & \enc{m(x).P} & \DEF & m(Y).Y\lrangle{\lrangle{x}\enc{P}} & \\ 
 & \enc{\overline{m}n.Q} &\DEF & \overline{m}[\lrangle{Z}(Z\lrangle{n})].\enc{Q} &  \\ 
\end{array}
\]
The encoding above uses both name parameterization and process parameterization. 
Typically one can assume that $Y$ and $Z$ are fresh for simplicity, but this is not essential, because these variables are bound and can be $\alpha$-converted whenever necessary, and moreover the encoded first-order process does not have higher-order variables.
Specifically, the encoding of an output `transmits' the name to be sent (i.e., $n$) in terms of a process parameterization (i.e., $\lrangle{Z}(Z\lrangle{n})$) that, once being received by the encoding of an input, is instantiated by a name-parameterized term (i.e., $\lrangle{x}\enc{P}$), which then can apply $n$ on $x$ in the encoding of $P$, thus fulfilling `name-passing'. Below we give an example. Suppose $P\DEF (c)(a(x).\overline{x}c.P_1)$ and $Q\DEF (d)(\overline{a}d.d(y).Q_1)$. So
\[
\begin{array}{lcl}
P\para Q &\st{\tau}& (d)((c)(\overline{d}c.P_1\fosub{d}{x})\para d(y).Q_1) \\
&\st{\tau}& (dc)(P_1\fosub{d}{x}\para Q_1\fosub{c}{y})
\end{array}
\] The encoding and interactions of $\enc{P\para Q}$ are as below. For clarity, we use \textbf{bold font} to indicate the evolving part during a communication.
\[
\begin{array}{lcl}
\enc{P\para Q} &\equiv& (c)(a(Y).Y\lrangle{\lrangle{x}\enc{\overline{x}c.P_1}}) \,\para\, (d)(\overline{a}[\lrangle{Z}(Z\lrangle{d})].\enc{d(y).Q_1}) \\
 &\st{\tau}& (d)\big((c)(\bm{(\lrangle{Z}(Z\lrangle{d}))\lrangle{\lrangle{x}\enc{\overline{x}c.P_1}}}) \,\para\, \enc{d(y).Q_1} \big) \\
 &\equiv& (d)\big((c)( \bm{\enc{\overline{x}c.P_1}\fosub{d}{x}}) \,\para\, \enc{d(y).Q_1} \big) \\
 &\equiv& (d)\big((c)( \bm{ (\overline{x}[\lrangle{Z}(Z\lrangle{c})].\enc{P_1}) \fosub{d}{x}}) \,\para\, d(Y).Y\lrangle{\lrangle{y}\enc{Q_1}}  \big) \\
 &\equiv& (d)\big((c)( \bm{ (\overline{d}[\lrangle{Z}(Z\lrangle{c})].\enc{P_1}\fosub{d}{x}) }) \,\para\, d(Y).Y\lrangle{\lrangle{y}\enc{Q_1}}  \big) \\
 &\st{\tau}&  (dc)\big(\enc{P_1}\fosub{d}{x} \,\para\, \bm{(\lrangle{Z}(Z\lrangle{c}))(\lrangle{\lrangle{y}\enc{Q_1}})} \big) \\
 &\equiv&  (dc)\big(\enc{P_1}\fosub{d}{x} \,\para\, \enc{Q_1}\fosub{c}{y} \big) \\
 &\equiv&  (dc)\big(\enc{P_1\fosub{d}{x}} \,\para\, \enc{Q_1\fosub{c}{y}} \big)
\end{array}
\]

Apparently the encoding is compositional, preserves the (free) names, and moreover divergence-reflecting (since the encoding does not introduce any extra internal action), as stated in the follow-up lemma whose proof is a standard induction. 
\begin{lemma}
Assume $P$ is a \FOPi\ process. The encoding above from \FOPi\ to \HOPiDd\ is compositional and divergence-reflecting; moreover $\enc{P}\fosub{n}{m} \equiv \enc{P\fosub{n}{m}}$.
\end{lemma}

\tdup{
\bc{TODO (across \HOPid\ and \FOPi):}
\begin{itemize}
\item \bc{Operational correspondence}. \bc{$\checkmark$}

\item \bc{Soundness/weak adequacy}.  \bc{$\checkmark$}
(Ground bisimilarity: $\WGB$; Strong ground bisimilarity: $\SGB$; Context bisimilarity: $\WCB$; Strong context bisimilarity: $\SCB$)
\[
P\WGB Q \mbox{~~ implies ~~} \enc{P} \WCB \enc{Q}
\]
\stress{(Possibly use the result in Section \ref{s:normal} on normal bisimulation.)}
\end{itemize}
}

\subsection{Operational correspondence}
We have the following properties clarifying the correspondence of actions before and after the encoding.
To delineate some case of the operational correspondence in terms of certain special input, i.e., a trigger, we define $\triggerD \DEF \lrangle{Z}\overline{m}Z$ in which $m$ is assumed to be fresh (it will also be used in Section \ref{s:normal}, but here simply allows for more flexible characterization of the operational correspondence). We note that sometimes existential quantification is omitted when it is clear from context.
\begin{lemma}\label{l:opcor}
Suppose $P$ is a \FOPi\ process.
(1) If $P \st{a(b)} P'$, then $\enc{P} \st{a(\lrangle{Z}(Z\lrangle{b}))} T$ and $T\SCB \enc{P'}$; ~
\tdup{
\item (\rc{useful?seems not! to remove!}) If $P \st{a(b)} P'$, then for some fresh $m$, $\enc{P} \st{a(\lrangle{Z}\overline{m}[Z\lrangle{b}])} T$ and $(m)(T \para !m(Y).Y) \WCB \enc{P'}$;
}
(2) If $P \st{a(b)} P'$, then $\enc{P} \st{a(\triggerD)} T$ and $(m)(T \para !m(Y).Y\lrangle{b}) \WCB \enc{P'}$; ~
(3) If $P \st{\overline{a}b} P'$, then $\enc{P} \st{\overline{a}[\lrangle{Z}(Z\lrangle{b})]} T$ and $T\SCB \enc{P'}$; ~
(4) If $P \st{\overline{a}(b)} P'$, then $\enc{P} \st{(b)\overline{a}[\lrangle{Z}(Z\lrangle{b})]} T$ and $T\SCB \enc{P'}$; ~
(5) If $P \st{\tau} P'$, then $\enc{P} \st{\tau} T$ and $T\SCB \enc{P'}$.
\end{lemma}

The converse is as below.
\begin{lemma}\label{l:opcor-conv}
Suppose $P$ is a \FOPi\ process.
(1) If $\enc{P} \st{a(\lrangle{Z}(Z\lrangle{b}))} T$, then $P \st{a(b)} P'$ and $T\SCB \enc{P'}$; ~
\tdup{
\item (\rc{useful? seems not! to remove!}) If for some fresh $m$, $\enc{P} \st{a(\lrangle{Z}\overline{m}[Z\lrangle{b}])} T$, then $P \st{a(b)} P'$ and $(m)(T \para !m(Y).Y) \WCB \enc{P'}$;
}
(2) If $\enc{P} \st{a(\triggerD)} T$, then $P \st{a(b)} P'$ and $(m)(T \para !m(Y).Y\lrangle{b}) \WCB \enc{P'}$; ~
(3) If $\enc{P} \st{\overline{a}[\lrangle{Z}(Z\lrangle{b})]} T$, then $P \st{\overline{a}b} P'$ and $T\SCB \enc{P'}$; ~
(4) If $\enc{P} \st{(b)\overline{a}[\lrangle{Z}(Z\lrangle{b})]} T$, then $P \st{\overline{a}(b)} P'$ and $T\SCB \enc{P'}$; ~
(5) If $\enc{P} \st{\tau} T$, then $P \st{\tau} P'$ and $T\SCB \enc{P'}$.
\end{lemma}

Lemma \ref{l:opcor} and Lemma \ref{l:opcor-conv} can be proven in a similar fashion 
\iftoggle{appendixing}{%
 (we give details in Appendix \ref{a:proofs-encoding}),
}{%
 (details can be found in \cite{Xu16app}),
}
and moreover be lifted to the weak situation. That is, if one replaces strong transitions (single arrows) with weak transitions (double arrows), 
the results still hold ($\SCB$ retains because the encoding does not bring any extra internal action); see \cite{San92,SW01a} for a reference.
We will however simply refer to these two lemmas in related discussions.

\subsection{Soundness}
In this section, we discuss the soundness of the encoding.
First of all, it is unfortunate that the soundness of the encoding is not true. To see this, take the processes $R_1$ and $R_2$ below. We recall that the CCS-like prefixes are defined as usual, i.e., $a.P\DEF a(x).P$ ($x\notin \n{P}$), $\overline{a}.P\DEF (c)\overline{a}c.P$ ($c\notin \n{P}$); sometimes we trim the trailing $0$, e.g., $a$ stands for $a.0$ and $\overline{a}$ for $\overline{a}.0$.
\[
\begin{array}{lcllcl}
R_1 &\DEF& (b)(a.\overline{b} \para b.\overline{c}) &\qquad\quad R_2 &\DEF& (b)(a.\overline{b} \para b.\overline{c} \para b.\overline{c})
\end{array}
\]
Obviously, $R_1$ and $R_2$ are ground bisimilar.
Now we examine their encodings. 
\[
\begin{array}{lcl}
\enc{R_1} &\equiv& (b)(a(Y).Y\lrangle{\lrangle{x}\enc{\overline{b}}} \para b(Y).Y\lrangle{\lrangle{x}\enc{\overline{c}}}) \\
\enc{R_2} &\equiv& (b)(a(Y).Y\lrangle{\lrangle{x}\enc{\overline{b}}} \para b(Y).Y\lrangle{\lrangle{x}\enc{\overline{c}}} \para b(Y).Y\lrangle{\lrangle{x}\enc{\overline{c}}})
\end{array}
\]

We show that $\enc{R_1}$ and $\enc{R_2}$ are not context bisimilar. Define $T\DEF (m)(\overline{a}[\lrangle{Z}\overline{m}Z] \para m(X).(X\lrangle{d} \para X\lrangle{d})$.
Then $(a)(\enc{R_1}\para T)$ and $(a)(\enc{R_2}\para T)$ can be distinguished. The latter can fire two output on $c$, whereas the former cannot, as shown below.
\[
\begin{array}{lrl}
 &(a)(\enc{R_1}\para T) \quad \st{\tau}\SCB& (m)((b)(\overline{m}[\lrangle{x}\enc{\overline{b}}] \para b(Y).Y\lrangle{\lrangle{x}\enc{\overline{c}}}) \para m(X).(X\lrangle{d} \para X\lrangle{d})) \\
&\st{\tau}\SCB& (b)(b(Y).Y\lrangle{\lrangle{x}\enc{\overline{c}}} \para \enc{\overline{b}} \para \enc{\overline{b}}) \\
&\equiv& (b)(b(Y).Y\lrangle{\lrangle{x}\enc{\overline{c}}} \para (e)\overline{b}[\lrangle{Z}(Z\lrangle{e})] \para \enc{\overline{b}}) \\
&\st{\tau}\SCB& (b)(\enc{\overline{c}} \para \enc{\overline{b}}) \\
&\equiv& (b)((f)\overline{c}[\lrangle{Z}(Z\lrangle{f})] \para \enc{\overline{b}}) \\
&\st{(f)\overline{c}[\lrangle{Z}(Z\lrangle{f})]}\SCB& 0 
\end{array}
\]
\[
\begin{array}{lrl}
 &(a)(\enc{R_2}\para T) \quad \st{\tau}\SCB& (m)((b)(\overline{m}[\lrangle{x}\enc{\overline{b}}] \para b(Y).Y\lrangle{\lrangle{x}\enc{\overline{c}}} \para b(Y).Y\lrangle{\lrangle{x}\enc{\overline{c}}}) \para m(X).(X\lrangle{d} \para X\lrangle{d})) \\
&\st{\tau}\SCB& (b)(b(Y).Y\lrangle{\lrangle{x}\enc{\overline{c}}}\para b(Y).Y\lrangle{\lrangle{x}\enc{\overline{c}}} \para \enc{\overline{b}} \para \enc{\overline{b}}) \\
&\equiv& (b)(b(Y).Y\lrangle{\lrangle{x}\enc{\overline{c}}}\para b(Y).Y\lrangle{\lrangle{x}\enc{\overline{c}}} \para (e)\overline{b}[\lrangle{Z}(Z\lrangle{e})] \para (e)\overline{b}[\lrangle{Z}(Z\lrangle{e})]) \\
&\st{\tau}\st{\tau}\SCB& \enc{\overline{c}} \para \enc{\overline{c}} \\
&\equiv& (f)\overline{c}[\lrangle{Z}(Z\lrangle{f})] \para (f)\overline{c}[\lrangle{Z}(Z\lrangle{f})] \\
&\st{(f)\overline{c}[\lrangle{Z}(Z\lrangle{f})]}\SCB& (f)\overline{c}[\lrangle{Z}(Z\lrangle{f})] \\
&\st{(f)\overline{c}[\lrangle{Z}(Z\lrangle{f})]}\SCB& 0
\end{array}
\]

Intuitively, the reason general soundness does not hold is that context bisimulation is somewhat more discriminating in the target higher-order calculus, which can have more flexibility when dealing with blocks of processes in presence of parameterization (e.g., some subprocess can be sent as needed). This is however beyond the capability of a first-order process.

In spite of the falsity of soundness in general, we can have a somewhat weaker yet still sensible soundness. Remember that our main goal is to achieve first-order concurrency in the higher-order model, so maybe we do not need to be so demanding when coping with the encodings of  first-order processes, that is, when testing an encoding process with an input, one can focus on those representing a name instead of a general one. Then it is expected that soundness will hold under this assumption. Fortunately, this is indeed true.



We have the following lemma stating the weak soundness of the encoding. Recall that $\WWCB$ is the $\WCB$ restricted to the image of the encoding (i.e., the processes in the target model that have reverse-image w.r.t. the encoding).
\begin{lemma}\label{l:soundness}
Suppose $P$ is a \FOPi\ process. Then $P\WGB Q$ implies $\enc{P} \,\WWCB\, \enc{Q}$.
\end{lemma}

\tdup{
$\myxcancel{
\fbox{
\begin{minipage}{8cm}
\rc{Use normal bisimulation for \HOPiDd\ to prove soundness, and original context bisimulation for completeness (?); Notice normal bisimulation for \HOPiDd\ inherits that for \HOPiD\ (type of input can be name-parameterization or process-parameterization.)}
\end{minipage}
}}$
}

\begin{proof}
\tdup{
\stress{ \scriptsize DONE! $\bcancel{\mbox{TODO: ,}}$
to deal with input and output, USE ``up-to context" technique (\cite{SW01a}, page 80-92; to confirm that it can be extended to higher-order paradigm (e.g., the case a context hole appears beneath an input or name-abstraction (seems ok)) !!; maybe also notice (e.g.) \cite{BPPR15}). }

\stress{ \scriptsize NOTICE that in input/output bisimulation (to prove $\enc{P}\WNB \enc{Q}$ or $\enc{P}\WCB \enc{Q}$):
\begin{itemize}
\item (by going through the FO processes $P,Q$) $\enc{P}\st{a(\triggerD)}$ must be able to be matched by $\enc{Q}\st{a(\triggerD)}$;
\item (by going through the FO processes $P,Q$)  $\enc{P}\st{\overline{a}[\lrangle{Z}(Z\lrangle{b})]}$ must be able to be matched by $\enc{Q}\st{\overline{a}[\lrangle{Z}(Z\lrangle{b})]}$.
\end{itemize}}
}


We show that $\mathcal{R}\DEF \{(\enc{P},\enc{Q}) \,|\, P\WGB Q\} \cup \WWCB$ is a context bisimulation up-to context and $\SCB$ (we refer the reader to, for example,  \cite{SW01a,BPPR15} and the references therein for the up-to proof technique for establishing bisimulations; we note that using $\SCB$ here is sufficient since it is stronger than $\WSCB$, i.e., $\SCB$ restricted to the image of the encoding). 

Suppose $\enc{P}\,\mathcal{R}\, \enc{Q}$. There are several cases, where  Lemma \ref{l:opcor} and Lemma \ref{l:opcor-conv} 
play an important part.
\begin{itemize}
\item $\enc{P}\wt{a(\lrangle{Z}(Z\lrangle{b}))} T$.
By Lemma \ref{l:opcor-conv}, $P \wt{a(b)} P'$ and $T\SCB \enc{P'}$. Because $P\WGB Q$, we know that $Q \wt{a(b)} Q'$ ~ $\WGB P'$ and thus $\enc{P'} \,\mathcal{R}\, \enc{Q'}$. Then by Lemma \ref{l:opcor}, $\enc{Q} \wt{a(\lrangle{Z}(Z\lrangle{b}))} T'$ and $T'\SCB \enc{Q'}$. So we have $T \SCB \enc{P'} \,\mathcal{R}\, \enc{Q'} \SCB T'$.

\tdup{
$\myxcancel{
\fbox{
\begin{minipage}{14cm}
$\enc{P}\wt{a(\triggerD)} T$. \stress{todo ~~ (\& see NOTICE just above: use Lemma \ref{l:opcor},~\ref{l:opcor-conv}(3) and up-to context}) \\
By Lemma \ref{l:opcor-conv}, $P \wt{a(b)} P'$ and $(m)(T \para !m(Y).Y\lrangle{b}) \WCB \enc{P'}$. Because $P\WGB Q$, we know that $Q \wt{a(b)} Q' \WGB P'$ and thus $\enc{P'} \,\mathcal{R}\, \enc{Q'}$. Then by Lemma \ref{l:opcor}, $\enc{Q} \wt{a(\triggerD)} T'$ and $(m)(T' \para !m(Y).Y\lrangle{b}) \WCB \enc{Q'}$. So  {\large \stress{??? (some informal discussion in ``q.txt")}} \\ 
\stress{Input is the crux!! Try ...
A compromise is to confine to $\enc{}(\FOPi)$
(i.e., the image of the encoding that communicate only $\lrangle{Z}(Z\lrangle{b})$)}; \\
\stress{under this constraint the argument for input would be straightforward (see the box on the right).}\\
Three possible motiv: 1) encoding used to achieve correct FO interactions in HO, so limited contexts may be ok; 2) too demanding to require all kinds of input because the target (HO) model has somewhat more powerful computation ability;  3) in practice usually not all possible objects are communicated (rather some typical ones).
\end{minipage}
}}$
}

\tdup{
$\myxcancel{
\fbox{
\begin{minipage}{14cm}
\rc{Some more discussion: tackle general input directly?}\\
{First let us consider a special case, i.e., the input is $\lrangle{Z}Z\lrangle{b}$, somewhat the encoding of a transmitted name.
By Lemma \ref{l:opcor-conv}, that $\enc{P}\wt{a(\lrangle{Z}(Z\lrangle{b}))} T_1$ implies that $P \wt{a(b)} P_1$ and $T_1\SCB \enc{P_1}$. Because $P\WGB Q$, we know that $Q \wt{a(b)} Q_1 \WGB P_1$ and thus $\enc{P_1} \,\mathcal{R}\, \enc{Q_1}$. Then by Lemma \ref{l:opcor}, $\enc{Q} \wt{a(\lrangle{Z}(Z\lrangle{b}))} T_2$ and $T_2\SCB \enc{Q_1}$. So we have $T_1 \SCB \enc{P_1} \,\mathcal{R}\, \enc{Q_1} \SCB T_2$. \\
Now consider the general input $A$, which basically should take the form $\lrangle{Z}F[Z\lrangle{b}]$ for some context $F$, so as to make the applications happen in a correct manner (otherwise the discussion would be similar, e.g., $Z$ does not appear in $F$ or is not fed with a name).
Say $\enc{P}\wt{a(\lrangle{Z}F[Z\lrangle{b}])} T$. Then we know from the special case above that $\enc{Q}\wt{a(\lrangle{Z}F[Z\lrangle{b}])} T'$.
\bc{The problem here is how to relate $T$ with $T_1$ (and $T'$ with $T_2$)}.
Specifically, we know $T_1\equiv G[(\lrangle{x}\enc{R})\lrangle{b}] \equiv G[\enc{R}\fosub{b}{x}]$ for some context $G$ and $\lrangle{x}R$. Then $T\equiv G[F[(\lrangle{x}\enc{R})\lrangle{b}]] \equiv G[F[\enc{R}\fosub{b}{x}]]$. Similarly, we have $T_2\equiv H[(\lrangle{x}\enc{R'})\lrangle{b}] \equiv H[\enc{R'}\fosub{b}{x}]$ for some context $H$ and $\lrangle{x}R'$, and $T'\equiv H[F[(\lrangle{x}\enc{R'})\lrangle{b}]] \equiv H[F[\enc{R'}\fosub{b}{x}]]$. The situation is depicted below.
\[
\xymatrix{
  T \ar@{}[r]|-{\equiv} & G[F[\enc{R}\fosub{b}{x}]] \ar@{.}[rr]|-{?}\ar@{.}[d]|-{?}  &  & H[F[\enc{R'}\fosub{b}{x}]] \ar@{.}[d]|-{?} \ar@{}[r]|-{\equiv}  & T'  \\
  T_1 \ar@{}[r]|-{\equiv} & G[\enc{R}\fosub{b}{x}] \ar@{}[rr]|-{ \SCB\, \enc{P_1}\,\mathcal{R}\,\enc{Q_1}\,\SCB}  & &  H[\enc{R'}\fosub{b}{x}] \ar@{}[r]|-{\equiv} & T_2
}
\]
\rc{How can we proceed? Use normal bisimulation, we can set $F$ as $\overline{m}[\cdot]$. But then how? }
}
\end{minipage}
}}$
}


\item $\enc{P}\wt{\overline{a}[(b)\lrangle{Z}(Z\lrangle{b})]} T$. 
By Lemma \ref{l:opcor-conv}, $P \wt{\overline{a}(b)} P'$ and $T\SCB \enc{P'}$. Because $P\WGB Q$, we know that $Q \wt{\overline{a}(b)} Q' \WGB P'$ and thus $\enc{P'} \,\mathcal{R}\, \enc{Q'}$. Then by Lemma \ref{l:opcor}, $\enc{Q} \st{(b)\overline{a}[\lrangle{Z}(Z\lrangle{b})]} T'$ and $T'\SCB \enc{Q'}$. Consider the following pair
\[
(b)(T\para E[A]) \;\quad,\quad\;  (b)(T'\para E[A])
\] in which $b\notin \fn{E[X]}$ 
and $A\DEF \lrangle{Z}(Z\lrangle{b})$. So
\[
(b)(T\para E[A]) \SCB (b)(\enc{P'}\para E[A]) \;\quad,\quad\; (b)(\enc{Q'}\para E[A]) \SCB (b)(T'\para E[A])
\] By setting a context $C\DEF (b)([\cdot]\para E[A])$, we have the following pair in which $\enc{P'} \,\mathcal{R}\, \enc{Q'}$.
\[
C[\enc{P'}] \;\quad,\quad\; C[\enc{Q'}]
\] This suffices to close this case in terms of the up-to context requirement.

\item $\enc{P}\wt{\overline{a}[\lrangle{Z}(Z\lrangle{b})]} T$. 
This case is similar to the last case.

\item $\enc{P}\wt{\tau} T$. By Lemma \ref{l:opcor-conv}, $P \wt{\tau} P'$ and $T\SCB \enc{P'}$. From $P\WGB Q$, we know $Q \wt{} Q'\WGB P'$ and thus $\enc{P'} \,\mathcal{R}\, \enc{Q'}$. Then by Lemma \ref{l:opcor}, $\enc{Q} \wt{} T'$ and $T'\SCB \enc{Q'}$. So we have $T\SCB \enc{P'} \,\mathcal{R}\, \enc{Q'}\SCB T'$.
\end{itemize}

\end{proof}

\subsection{Completeness}
The completeness of the encoding is stated in the lemma below. {We note that completeness is true even if we do not constrain the domain to be the image of the encoded \FOPi\ processes. }
\begin{lemma}\label{l:completeness}
Suppose $P$ is a \FOPi\ process. Then $\enc{P} \WCB \enc{Q}$ implies $P\WGB Q$.
\end{lemma}
\begin{proof}
\tdup{
\stress{ \scriptsize DONE! $\bcancel{\mbox{TODO: if necessary,}}$ \\
(1) (for bound output in particular) maybe Use ``local bisimulation" on the \FOPi\ side and context bisimulation on \HOPiDd\ side, to analyze using ``context surjection" (e.g., \cite{Fu05b,Xu12}). \\
(2) maybe use certain up-to technique on the FO side.}
}

\tdup{
\stress{\scriptsize NOTICE that in input/output bisimulation (to prove $P\WGB Q$):
\begin{itemize}
\item to prove $P$ and $Q$ bisimulate on input: (by going through the HO processes $\enc{P},\enc{Q}$) $\enc{P}\st{a(\lrangle{Z}(Z\lrangle{b}))}$ must be able to be matched by $\enc{Q}\st{a(\lrangle{Z}(Z\lrangle{b}))}$;
\item to prove $P$ and $Q$ bisimulate on free output:(by going through the HO processes $\enc{P},\enc{Q}$)  $\enc{P}\st{\overline{a}[\lrangle{Z}(Z\lrangle{b})]}$ must be able to be matched by $\enc{Q}\st{\overline{a}[\lrangle{Z}(Z\lrangle{b})]}$, because $\enc{Q}$ can only emit such form of processes, and moreover if the matching is $\enc{Q}\st{\overline{a}[\lrangle{Z}(Z\lrangle{c})]}$ then one can design a context to distinguish $\enc{P}$ and $\enc{Q}$;
\item to prove $P$ and $Q$ bisimulate on bound output:(by going through the HO processes $\enc{P},\enc{Q}$)  $\enc{P}\st{\overline{a}[(b)\lrangle{Z}(Z\lrangle{b})]}$ must be able to be matched by $\enc{Q}\st{(b)\overline{a}[\lrangle{Z}(Z\lrangle{b})]}$ (apply $\alpha$-conversion if needed), because $\enc{Q}$ can only emit such form of processes, and moreover if the matching is $\enc{Q}\st{(c)\overline{a}[\lrangle{Z}(Z\lrangle{c})]}$ (or $\enc{Q}\st{\overline{a}[\lrangle{Z}(Z\lrangle{c})]}$) then one can design a context to distinguish between $\enc{P}$ and $\enc{Q}$;
\end{itemize}}
}

We show that $\mathcal{R}\DEF \{(P,Q) \,|\, \enc{P} \WCB \enc{Q}\} \cup \WGB$ is a local bisimulation. 
Suppose $P\,\mathcal{R}\, Q$. There are several cases. 
\begin{itemize}
\item $P\wt{a(b)} P'$. 
By Lemma \ref{l:opcor}, $\enc{P} \wt{a(\lrangle{Z}(Z\lrangle{b}))} T$ and $T\SCB \enc{P'}$. Because $\enc{P} \WCB \enc{Q}$, we know that $\enc{Q} \wt{a(\lrangle{Z}(Z\lrangle{b}))} T'\WCB T$. By Lemma \ref{l:opcor-conv}, $Q \wt{a(b)} Q'$ and $T'\SCB \enc{Q'}$. Thus we have $\enc{P'} \SCB T \WCB T'\SCB \enc{Q'}$, so $P'\,\mathcal{R}\, Q'$, which fulfills this case.

\item $P\wt{\overline{a}b} P'$. 
By Lemma \ref{l:opcor}, $\enc{P} \wt{\overline{a}[\lrangle{Z}(Z\lrangle{b})]} T$ and $T\SCB \enc{P'}$. Since $\enc{P} \WCB \enc{Q}$, we know that $\enc{P}$ must be able to be matched by $\enc{Q}\wt{\overline{a}[\lrangle{Z}(Z\lrangle{b})]} T'$, because $\enc{Q}$ can only output such shape of processes, and if the matching is, e.g., $\enc{Q}\wt{\overline{a}[\lrangle{Z}(Z\lrangle{c})]} T''$ then a context can be designed to distinguish between $\enc{P}$ and $\enc{Q}$. So for every $E[X]$, we have $T \para E[\lrangle{Z}(Z\lrangle{b})] \WCB T' \para E[\lrangle{Z}(Z\lrangle{b})]$. By Lemma \ref{l:opcor-conv}, $Q \wt{\overline{a}b} Q'$ and $T'\SCB \enc{Q'}$. So we know
\begin{equation}\label{eq:comp4}
\enc{P'} \para E[\lrangle{Z}(Z\lrangle{b})) \WCB \enc{Q'} \para E[\lrangle{Z}(Z\lrangle{b})]
\end{equation}
We want to show
\begin{equation}\label{eq:comp5}
P' \mathcal{R} Q' \quad \mbox{ that is, }\; \enc{P'} \WCB \enc{Q'}
\end{equation}
By setting $E$ to be $0$ in (\ref{eq:comp4}), we obtain (\ref{eq:comp5}), and thus close this case.

\item $P\wt{\overline{a}(b)} P'$. 
By Lemma \ref{l:opcor}, $\enc{P} \wt{(b)\overline{a}[\lrangle{Z}(Z\lrangle{b})]} T$ and $T\SCB \enc{P'}$. Since $\enc{P} \WCB \enc{Q}$, we know that $\enc{P}$  must be able to be matched by $\enc{Q}\wt{(b)\overline{a}[\lrangle{Z}(Z\lrangle{b})]} T'$ (apply $\alpha$-conversion if needed). This is because $\enc{Q}$ can only emit such form of processes, and moreover if the matching does not have a bound name (e.g., $\enc{Q}\wt{\overline{a}[\lrangle{Z}(Z\lrangle{c})]} T''$) then one can design a context to distinguish $\enc{P}$ and $\enc{Q}$. So for every $E[X]$ s.t. $b\notin \mbox{fn}(E)$, we have $(b)(T \para E[\lrangle{Z}(Z\lrangle{b})]) \WCB (b)(T' \para E[\lrangle{Z}(Z\lrangle{b})])$. By Lemma \ref{l:opcor-conv}, $Q \wt{\overline{a}(b)} Q'$ and $T'\SCB \enc{Q'}$. So we know
\begin{equation}\label{eq:comp1}
(b)(\enc{P'} \para E[\lrangle{Z}(Z\lrangle{b})]) \WCB (b)(\enc{Q'} \para E[\lrangle{Z}(Z\lrangle{b})])
\end{equation}
In terms of local bisimulation \cite{Fu05b,Xu12}, for every \FOPi\ process $R$, we need to show
\begin{equation}\label{eq:comp2}
(b)(P'\para R) \,\mathcal{R}\, (b)(Q'\para R) \quad\mbox{ i.e., }\quad (b)(\enc{P'}\para \enc{R}) \WCB (b)(\enc{Q'} \para \enc{R})
\end{equation}
Comparing equations (\ref{eq:comp1}) and (\ref{eq:comp2}), one can see that the different part is $E[\lrangle{Z}(Z\lrangle{b})]$ and $\enc{R}$. Since the inverse of the encoding is a surjection, if all possible forms of $E$ is itinerated, $\enc{R}$ must be hit somewhere (i.e., some choice of $E$ makes $E[\lrangle{Z}(Z\lrangle{b})]$ and $\enc{R}$ equal). Therefore we infer that (\ref{eq:comp2}) 
is true and thus complete this case.

\item $P\wt{\tau} P'$. By Lemma \ref{l:opcor}, $\enc{P} \wt{\tau} T$ and $T\SCB \enc{P'}$. Because $\enc{P} \WCB \enc{Q}$, we know $\enc{Q} \wt{} T' \WCB T$. Then by Lemma \ref{l:opcor-conv}, $Q \wt{} Q'$ and $T'\SCB \enc{Q'}$. 
So we have $P'\,\mathcal{R}\, Q'$ because $\enc{P'}\SCB T \WCB T' \SCB \enc{Q'}$.
\end{itemize}

\end{proof}


%% file: normal.tex
\section{Normal bisimulation for \HOPiDd}\label{s:normal}
In this section, we show that context bisimulation in \HOPiDd\ can be characterized by the much simpler normal bisimulation. 

\tdup{
\bc{TODO (in \HOPiDd):}
\begin{itemize}
\item Factorization theorem; \bc{$\checkmark$}
\item Definition of normal bisimulation $\WNB$; \bc{$\checkmark$}
\item Coincidence between $\WNB$ and $\WCB$ (i.e., $\WNB$ implies $\WCB$ using the Factorization theorem). \bc{$\checkmark$}
\end{itemize}
\sep\sep

\fbox{We stipulate that $\rc{\trigger} \DEF \overline{m}$, $\rc{\triggerD} \DEF \lrangle{Z}\overline{m}Z$, and $\rc{\triggerd} \DEF \lrangle{z}\overline{m}[\lrangle{Y}(Y\lrangle{z})]$.
}
}

\subsubsection*{The factorization theorem}
Below is the factorization theorem in presence of parameterization on names (and on processes as well). We recall that $\equiv$ is the structural congruence.
As explained in Section \ref{s:introduction}, the upshot of establishing the factorization theorem is to find the right small processes so-called triggers. Here we have three kinds of triggers, to tackle different kinds of parameterizations. In particular, we stipulate that the triggers are as follows:  $\rrc{\triggerd} \DEF \lrangle{z}\overline{m}[\lrangle{Y}(Y\lrangle{z})]$, $\rrc{\triggerD} \DEF \lrangle{Z}\overline{m}Z$, and $\rrc{\trigger} \DEF \overline{m}$. These triggers are of somewhat a similar flavor but quite different in shape, with the aim at factorizing out respectively a name abstraction, a process abstraction and a non-abstraction process in certain context. The first trigger, i.e.,  $\triggerd$, is the main contribution of this work, whereas the other two are inherited from \cite{Xu13} and \cite{San92} respectively.
\begin{theorem}[Factorization]\label{factor-bigd-smalld} 
Given $E[X]$ of \HOPiDd, it holds for every $A$, fresh $m$ (i.e., $m\notin fn(E,A)$) that
\begin{itemize}
\item[(1)] if $E[X]$ is not an abstraction, then
\begin{itemize}
\item[(i)] if $A$ is not an abstraction, then
$E[A] \WCB (m)(E[\trigger] \para  !m.A)
$;
\item[(ii)] if $A$ is an abstraction on process, then
$E[A] \WCB (m)(E[\triggerD] \para  !m(Z).A\lrangle{Z})
$;
\item[(iii)] if $A$ is an abstraction on name, then
$E[A] \WCB (m)(E[\triggerd] \para  !m(Z).Z\lrangle{A})
$.
\end{itemize}

\item[(2)]  else if $E[X]$ is an abstraction, i.e., $E[X]\equiv \ve{\lrangle{U}}E'$ for some non-abstraction $E'$ (here $\ve{\lrangle{U}}$ denotes the abstractions prefixing $E'$), 
then
\begin{itemize}
\item[(i)] if $A$ is not an abstraction, then
$E[A] \WCB \ve{\lrangle{U}}((m)(E'[\trigger] \para  !m.A))
$;
\item[(ii)] if $A$ is an abstraction on process, then
$E[A] \WCB \ve{\lrangle{U}}((m)(E'[\triggerD] \para  !m(Z).A\lrangle{Z}))
$;
\item[(iii)] if $A$ is an abstraction on name, then
$E[A] \WCB \ve{\lrangle{U}}((m)(E'[\triggerd] \para  !m(Z).Z\lrangle{A}))
$.
\end{itemize}


\end{itemize}
\end{theorem}
In Theorem \ref{factor-bigd-smalld}, the clause (i) of (1) and (2) is actually Sangiorgi's seminal work \cite{San92}. The clause (ii) of (1) and (2) is analyzed in \cite{Xu13}. The clause (iii) of (1) and (2), which depicts the factorization for abstraction on names, can be discussed through a technical routine almost the same as (ii). 
With regard to more details we refer the reader to \cite{San92, SW01a, Xu13}.

The method of \emph{trigger} (including the technical approach) is well-developed in the field, due to the fundamental framework by Sangiorgi \cite{SW01a}. So the key  to establishing factorization for processes allowing abstraction on names is the \emph{trigger}, which is not known for a long time in contrast to the cases of abstraction on processes and that without abstractions. Once a right trigger is found, the rest of discussion is then almost standard. Below we give an example of the factorization concerning abstraction on names. 

\noindent\textbf{Example} The basic idea of factorization concerning abstraction on names can be illustrated in the following example in which $m$ is fresh (i.e., not in $A\lrangle{\rrc{d}}$).
\begin{eqnarray}
A\lrangle{\rrc{d}} &\WCB& (m)(\; (\lrangle{\rrc{z}}\overline{m}[\rbc{\lrangle{Y}(Y\lrangle{z})}])\lrangle{\rrc{d}} \,\para\, m(\rbc{Z}).\rbc{Z}\lrangle{A} \;) 
\; \equiv\; (m)(\; \overline{m}[\rbc{\lrangle{Y}(Y\lrangle{\rrc{d}})}] \,\para\, m(\rbc{Z}).\rbc{Z}\lrangle{A} \;) \nonumber
\end{eqnarray}
\noindent For example, if $A$ is $\lrangle{x}\overline{x}b$, then $A\lrangle{d} \equiv \overline{d}b$, and 
\[
\begin{array}{lclcl}
A\lrangle{d}  &\WCB& (m)(\; \overline{m}[\lrangle{Y}(Y\lrangle{d})] \,\para\, m(Z).Z\lrangle{A} \;) 
 &\WCB& (m)(\; (\lrangle{Y}(Y\lrangle{d}))\lrangle{A} \;) \;\equiv\; A\lrangle{d} \;\equiv\; \overline{d}b 
\end{array}
\]

\subsubsection*{Normal bisimulation for \HOPiDd}
Below is the definition of normal bisimulation whose clauses are designed with regard to the factorization theorem. We recall that $\rrc{\trigger} \DEF \overline{m}$, $\rrc{\triggerD} \DEF \lrangle{Z}\overline{m}Z$, and $\rrc{\triggerd} \DEF \lrangle{z}\overline{m}[\lrangle{Y}(Y\lrangle{z})]$.
\begin{definition}\label{normal-bisi-Dd} 
A symmetric binary relation $\mathcal{R}$ on closed processes of \HOPiDd\ is a normal bisimulation, if whenever $P\,\mathcal{R}\, Q$ the following properties hold:
\begin{enumerate}
\item If $P \st{a(\trigger)} P'$ ($m$ is fresh w.r.t. $P$ and $Q$), then $Q \wt{a(\trigger)} Q'$ for some $Q'$ s.t.  $P'\,\mathcal{R}\, Q'$;
\item If $P \st{a(\triggerD)} P'$ ($m$ is fresh w.r.t. $P$ and $Q$), then $Q \wt{a(\triggerD)} Q'$ for some $Q'$ s.t.  $P'\,\mathcal{R}\, Q'$;
\item If $P \st{a(\triggerd)} P'$ ($m$ is fresh w.r.t. $P$ and $Q$), then $Q \wt{a(\triggerd)} Q'$ for some $Q'$ s.t.  $P'\,\mathcal{R}\, Q'$;

\item If $P \st{(\ve{c})\overline{a}A} P'$ and $A$ is not an abstraction, then $Q \wt{(\ve{d})\overline{a}B} Q'$ for some $\ve{d},Q'$ and $B$ that is not an abstraction, and it holds that ($m$ is fresh) ~
$(\ve{c})(P'\para !\rbc{m.A}) \; \mathcal{R}\;  (\ve{d})(Q'\para  !\rbc{m.B})$.
\item If $P \st{(\ve{c})\overline{a}A} P'$ and $A$ is an abstraction on process, then $Q \wt{(\ve{d})\overline{a}B} Q'$ for some $\ve{d},Q'$ and $B$ that is an abstraction on process, and it holds that ($m$ is fresh) ~
$(\ve{c})(P'\para !\rbc{m(Z).A\lrangle{Z}}) \; \mathcal{R}\;  (\ve{d})(Q'\para  !\rbc{m(Z).B\lrangle{Z}})$.
\item If $P \st{(\ve{c})\overline{a}A} P'$ and $A$ is an abstraction on name, then $Q \wt{(\ve{d})\overline{a}B} Q'$ for some $\ve{d},Q'$ and $B$ that is an abstraction on name, and it holds that ($m$ is fresh) ~
$(\ve{c})(P'\para !\rbc{m(Z).Z\lrangle{A}}) \; \mathcal{R}\;  (\ve{d})(Q'\para  !\rbc{m(Z).Z\lrangle{B}})$.

\item If $P \st{\tau} P'$, then $Q \wt{} Q'$ for some $Q'$ s.t. $P'\,\mathcal{R}\, Q'$;
\end{enumerate}
Process $P$ is normal bisimilar to $Q$, written $P\,\WNB\, Q$, if $P\,\mathcal{R}\, Q$ for some normal bisimulation $\mathcal{R}$. Relation \WNB\ is called normal bisimilarity, and is a congruence (see \cite{San92} for a reference). The strong version of \WNB\ is denoted by \SNB. 
\end{definition}

\subsubsection*{Coincidence between normal bisimilarity and context bisimilarity in \HOPiDd}
Now we have the following theorem. 
\iftoggle{appendixing}{%
 The proof is in Appendix \ref{a:proofs-normal}.
}{%
 The detailed proof is referred to \cite{Xu16app}.
}
\begin{theorem}\label{normal-characterization-hopiDd} 
In \HOPiDd, normal bisimilarity coincides with context bisimilarity; that is, $\WNB \,=\, \WCB$.
\end{theorem}


%% file: conclusion.tex
\section{Conclusion}\label{s:conclusion}
In this paper, we have exhibited a new encoding of name-passing in the higher-order paradigm that allows parameterization, and a normal bisimulation in that setting as well. In the former, we demonstrate the conformance of the encoding to the well-established criteria in the literature. In the latter, we prove the coincidence between normal and context bisimulation by pinpointing how to factorize an abstraction on some name. The encoding of this work is inspired by the one proposed by Alan Schmitt during the communication concerning another work. That encoding, as given below, somewhat swaps the roles of input and output and treats $a(x).P$ somehow as $a.\lrangle{x}P$ (like those calculi admitting abstractions and concretions \cite{San92}). 
\[
\begin{array} {rcl}
\enc{a(x).P} & \DEF & \overline{a}[\lrangle{x}\enc{P}]\\
\enc{\overline{a}b.Q} &\DEF & a(Y).(Y\lrangle{b}\para \enc{Q}) 
\end{array}
\]
From the angle of achieving first-order interaction, the encoding strategy above is truly interesting. However, it appears not to satisfy some usual operational correspondence (say, in \cite{Gor08a} or \cite{LPSS10}), and full abstraction is not quite clear. 
Based on the results in this paper, it is tempting to expect that this encoding have some (nearly) same properties, and this is worthwhile for more investigation.

The results of this paper can be dedicated to facilitate further study on the expressiveness of higher-order processes. The following questions, among others, are still open: whether \FOPi\ can be encoded in a higher-order setting only allowing parameterization on processes; whether there is a better encoding of \FOPi\ than the one in \cite{XYL15}, using higher-order processes only capable of parameterization on names; whether \HOPid\ afford a normal-like characterization of context bisimulation.

%% file: main.bbl
\begin{thebibliography}{10}
\providecommand{\bibitemdeclare}[2]{}
\providecommand{\surnamestart}{}
\providecommand{\surnameend}{}
\providecommand{\urlprefix}{Available at }
\providecommand{\url}[1]{\texttt{#1}}
\providecommand{\href}[2]{\texttt{#2}}
\providecommand{\urlalt}[2]{\href{#1}{#2}}
\providecommand{\doi}[1]{doi:\urlalt{http://dx.doi.org/#1}{#1}}
\providecommand{\bibinfo}[2]{#2}

\bibitemdeclare{inproceedings}{BPPR15}
\bibitem{BPPR15}
\bibinfo{author}{F.~\surnamestart Bonchi\surnameend},
  \bibinfo{author}{D.~\surnamestart Petrisan\surnameend},
  \bibinfo{author}{D.~\surnamestart Pous\surnameend} \&
  \bibinfo{author}{J.~\surnamestart Rot\surnameend} (\bibinfo{year}{2015}):
  \emph{\bibinfo{title}{Lax Bialgebras and Up-To Techniques for Weak
  Bisimulations}}.
\newblock In: {\sl \bibinfo{booktitle}{Proceedings of the 26th International
  Conference on Concurrency Theory (CONCUR 2015)}}, {\sl
  \bibinfo{series}{Leibniz International Proceedings in Informatics
  (LIPICS)}}~\bibinfo{volume}{42}, pp. \bibinfo{pages}{240--253},
  \doi{10.4230/LIPIcs.CONCUR.2015.240}.

\bibitemdeclare{article}{BHG06}
\bibitem{BHG06}
\bibinfo{author}{M.~\surnamestart Bundgaard\surnameend},
  \bibinfo{author}{T.~\surnamestart Hildebrandt\surnameend} \&
  \bibinfo{author}{J.~C. \surnamestart Godskesen\surnameend}
  (\bibinfo{year}{2006}): \emph{\bibinfo{title}{A CPS Encoding of Name-passing
  in Higher-order Mobile Embedded Resources}}.
\newblock {\sl \bibinfo{journal}{Theoretical Computer Science}}
  \bibinfo{volume}{356(3)}, pp. \bibinfo{pages}{422--439},
  \doi{10.1016/j.tcs.2006.02.006}.

\bibitemdeclare{techreport}{EN86a}
\bibitem{EN86a}
\bibinfo{author}{U.~H. \surnamestart Engberg\surnameend} \&
  \bibinfo{author}{M.~\surnamestart Nielsen\surnameend} (\bibinfo{year}{1986}):
  \emph{\bibinfo{title}{A Calculus of Communicating Systems with Label
  Passing}}.
\newblock \bibinfo{type}{Technical Report} \bibinfo{number}{DAIMI PB-208},
  \bibinfo{institution}{Computer Science Department, University of Aarhus}.
\newblock \urlprefix\url{http://www.daimi.au.dk/PB/208/}.

\bibitemdeclare{incollection}{EN00}
\bibitem{EN00}
\bibinfo{author}{U.~H. \surnamestart Engberg\surnameend} \&
  \bibinfo{author}{M.~\surnamestart Nielsen\surnameend} (\bibinfo{year}{2000}):
  \emph{\bibinfo{title}{A Calculus of Communicating Systems with Label Passing
  - Ten Years After}}.
\newblock In: {\sl \bibinfo{booktitle}{Proof, Language, and Interaction: Essays
  in Honour of Robin Milner}}, \bibinfo{publisher}{MIT Press Cambridge}, pp.
  \bibinfo{pages}{599--622}.

\bibitemdeclare{article}{Fu05b}
\bibitem{Fu05b}
\bibinfo{author}{Yuxi \surnamestart Fu\surnameend} (\bibinfo{year}{2005}):
  \emph{\bibinfo{title}{On Quasi Open Bisimulation}}.
\newblock {\sl \bibinfo{journal}{Theoretical Computer Science}}
  \bibinfo{volume}{338(1-3)}, pp. \bibinfo{pages}{96--126},
  \doi{10.1016/j.tcs.2004.10.041}.

\bibitemdeclare{article}{Fu15}
\bibitem{Fu15}
\bibinfo{author}{Yuxi \surnamestart Fu\surnameend} (\bibinfo{year}{2015}):
  \emph{\bibinfo{title}{Theory of interaction}}.
\newblock {\sl \bibinfo{journal}{Theoretical Computer Science}},
  \doi{10.1016/j.tcs.2015.07.043}.

\bibitemdeclare{article}{FL09a}
\bibitem{FL09a}
\bibinfo{author}{Yuxi \surnamestart Fu\surnameend} \& \bibinfo{author}{Hao
  \surnamestart Lu\surnameend} (\bibinfo{year}{2010}): \emph{\bibinfo{title}{On
  the Expressiveness of Interaction}}.
\newblock {\sl \bibinfo{journal}{Theoretical Computer Science}}
  \bibinfo{volume}{411}, pp. \bibinfo{pages}{1387--1451},
  \doi{10.1016/j.tcs.2009.11.011}.

\bibitemdeclare{inproceedings}{Gor08a}
\bibitem{Gor08a}
\bibinfo{author}{D.~\surnamestart Gorla\surnameend} (\bibinfo{year}{2008}):
  \emph{\bibinfo{title}{Towards a Unified Approach to Encodability and
  Separation Results for Process Calculi}}.
\newblock In: {\sl \bibinfo{booktitle}{Proceedings of the 19th International
  Conference on Concurrency Theory (CONCUR 2008)}}, {\sl
  \bibinfo{series}{LNCS}} \bibinfo{volume}{5201}, \bibinfo{publisher}{Springer
  Verlag}, pp. \bibinfo{pages}{492--507}, \doi{10.1007/978-3-540-85361-9\_38}.

\bibitemdeclare{article}{GN16}
\bibitem{GN16}
\bibinfo{author}{D.~\surnamestart Gorla\surnameend} \&
  \bibinfo{author}{U.~\surnamestart Nestmann\surnameend}
  (\bibinfo{year}{2016}): \emph{\bibinfo{title}{Full Abstraction for
  Expressiveness: History, Myths and Facts}}.
\newblock {\sl \bibinfo{journal}{Mathematical Structures in Computer Scinece}}
  \bibinfo{volume}{26}, pp. \bibinfo{pages}{639--654},
  \doi{10.1017/S0960129514000279}.

\bibitemdeclare{article}{Gor09}
\bibitem{Gor09}
\bibinfo{author}{Daniele \surnamestart Gorla\surnameend}
  (\bibinfo{year}{2009}): \emph{\bibinfo{title}{On the Relative Expressive
  Power of Calculi for Mobility}}.
\newblock {\sl \bibinfo{journal}{Electronic Notes in Theoretical Computer
  Science}} \bibinfo{volume}{249}, pp. \bibinfo{pages}{269--286},
  \doi{10.1016/j.entcs.2009.07.094}.

\bibitemdeclare{inproceedings}{KPY16}
\bibitem{KPY16}
\bibinfo{author}{D.~\surnamestart Kouzapas\surnameend}, \bibinfo{author}{J.~A.
  \surnamestart P\'{e}rez\surnameend} \& \bibinfo{author}{Nobuko \surnamestart
  Yoshida\surnameend} (\bibinfo{year}{2016}): \emph{\bibinfo{title}{On the
  Relative Expressiveness of Higher-Order Session Processes}}.
\newblock In: {\sl \bibinfo{booktitle}{Proceedings of the 25th European
  Symposium on Programming (ESOP 2016)}}, \bibinfo{series}{LNCS}, pp.
  \bibinfo{pages}{446--475}, \doi{10.1007/978-3-662-49498-1$\_$18}.

\bibitemdeclare{inproceedings}{LPSS10}
\bibitem{LPSS10}
\bibinfo{author}{I.~\surnamestart Lanese\surnameend}, \bibinfo{author}{J.~A.
  \surnamestart P\'{e}rez\surnameend}, \bibinfo{author}{D.~\surnamestart
  Sangiorgi\surnameend} \& \bibinfo{author}{A.~\surnamestart
  Schmitt\surnameend} (\bibinfo{year}{2010}): \emph{\bibinfo{title}{On the
  Expressiveness of Polyadic and Synchronous Communication in Higher-Order
  Process Calculi}}.
\newblock In: {\sl \bibinfo{booktitle}{Proceedings of the 36th International
  Colloquium on Automata, Languages and Programming (ICALP 2010)}},
  \bibinfo{series}{LNCS}, \bibinfo{publisher}{Springer Verlag}, pp.
  \bibinfo{pages}{442--453}, \doi{10.1007/978-3-642-14162-1\_37}.

\bibitemdeclare{inproceedings}{LPSS08}
\bibitem{LPSS08}
\bibinfo{author}{I.~\surnamestart Lanese\surnameend}, \bibinfo{author}{J.A.
  \surnamestart P\'{e}rez\surnameend}, \bibinfo{author}{D.~\surnamestart
  Sangiorgi\surnameend} \& \bibinfo{author}{A.~\surnamestart
  Schmitt\surnameend} (\bibinfo{year}{2008}): \emph{\bibinfo{title}{On the
  Expressiveness and Decidability of Higher-Order Process Calculi}}.
\newblock In: {\sl \bibinfo{booktitle}{Proceedings of the 23rd Annual IEEE
  Symposium on Logic in Computer Science (LICS 2008)}},
  \bibinfo{publisher}{IEEE Computer Society}, pp. \bibinfo{pages}{145--155},
  \doi{10.1109/LICS.2008.8}.
\newblock \bibinfo{note}{Journal version in \cite{LPSS10a}}.

\bibitemdeclare{inproceedings}{LSS09}
\bibitem{LSS09}
\bibinfo{author}{S.~\surnamestart Lenglet\surnameend},
  \bibinfo{author}{A.~\surnamestart Schmitt\surnameend} \&
  \bibinfo{author}{J.-B. \surnamestart Stefani\surnameend}
  (\bibinfo{year}{2009}): \emph{\bibinfo{title}{Normal Bisimulations in Calculi
  with Passivation}}.
\newblock In: {\sl \bibinfo{booktitle}{Proceedings of the 12th International
  Conference on Foundations of Software Science and Computational Structures
  (FOSSACS 2009)}}, {\sl \bibinfo{series}{LNCS}} \bibinfo{volume}{5504},
  \bibinfo{publisher}{Springer Verlag}, pp. \bibinfo{pages}{257--271},
  \doi{10.1007/978-3-642-00596-1\_19}.

\bibitemdeclare{article}{LSS11}
\bibitem{LSS11}
\bibinfo{author}{S.~\surnamestart Lenglet\surnameend},
  \bibinfo{author}{A.~\surnamestart Schmitt\surnameend} \&
  \bibinfo{author}{J.-B. \surnamestart Stefani\surnameend}
  (\bibinfo{year}{2011}): \emph{\bibinfo{title}{Characterizing Contextual
  Equivalence in Calculi with Passivation}}.
\newblock {\sl \bibinfo{journal}{Information and Computation}}
  \bibinfo{volume}{209}, pp. \bibinfo{pages}{1390--1433},
  \doi{10.1016/j.ic.2011.08.002}.

\bibitemdeclare{book}{Mil89}
\bibitem{Mil89}
\bibinfo{author}{R.~\surnamestart Milner\surnameend} (\bibinfo{year}{1989}):
  \emph{\bibinfo{title}{Communication and Concurrency}}.
\newblock \bibinfo{publisher}{Prentice Hall}.

\bibitemdeclare{article}{MPW92}
\bibitem{MPW92}
\bibinfo{author}{R.~\surnamestart Milner\surnameend},
  \bibinfo{author}{J.~\surnamestart Parrow\surnameend} \&
  \bibinfo{author}{D.~\surnamestart Walker\surnameend} (\bibinfo{year}{1992}):
  \emph{\bibinfo{title}{A Calculus of Mobile Processes (Parts I and II)}}.
\newblock {\sl \bibinfo{journal}{Information and Computation}}
  \bibinfo{volume}{100(1)}, pp. \bibinfo{pages}{1--77},
  \doi{10.1016/0890-5401(92)90008-4, 10.1016/0890-5401(92)90009-5}.

\bibitemdeclare{article}{Par16}
\bibitem{Par16}
\bibinfo{author}{J.~\surnamestart Parrow\surnameend} (\bibinfo{year}{2016}):
  \emph{\bibinfo{title}{General Conditions for Full Abstraction}}.
\newblock {\sl \bibinfo{journal}{Mathematical Structures in Computer Science}}
  \bibinfo{volume}{26}, pp. \bibinfo{pages}{655--657},
  \doi{10.1017/s0960129514000280}.

\bibitemdeclare{phdthesis}{San92}
\bibitem{San92}
\bibinfo{author}{D.~\surnamestart Sangiorgi\surnameend} (\bibinfo{year}{1992}):
  \emph{\bibinfo{title}{Expressing Mobility in Process Algebras: First-order
  and Higher-order Paradigms}}.
\newblock \bibinfo{type}{Phd thesis}, \bibinfo{school}{University of
  Edinburgh}.

\bibitemdeclare{article}{San94}
\bibitem{San94}
\bibinfo{author}{D.~\surnamestart Sangiorgi\surnameend} (\bibinfo{year}{1996}):
  \emph{\bibinfo{title}{Bisimulation for Higher-order Process Calculi}}.
\newblock {\sl \bibinfo{journal}{Information and Computation}}
  \bibinfo{volume}{131(2)}, pp. \bibinfo{pages}{141--178},
  \doi{10.1006/inco.1996.0096}.

\bibitemdeclare{article}{San98}
\bibitem{San98}
\bibinfo{author}{D.~\surnamestart Sangiorgi\surnameend} (\bibinfo{year}{1998}):
  \emph{\bibinfo{title}{On the Bisimulation Proof Method}}.
\newblock {\sl \bibinfo{journal}{Mathematical Structures in Computer Science}}
  \bibinfo{volume}{8(6)}, pp. \bibinfo{pages}{447--479},
  \doi{10.1017/S0960129598002527}.

\bibitemdeclare{book}{SW01a}
\bibitem{SW01a}
\bibinfo{author}{D.~\surnamestart Sangiorgi\surnameend} \&
  \bibinfo{author}{D.~\surnamestart Walker\surnameend} (\bibinfo{year}{2001}):
  \emph{\bibinfo{title}{The Pi-calculus: a Theory of Mobile Processes}}.
\newblock \bibinfo{publisher}{Cambridge Universtity Press}.

\bibitemdeclare{phdthesis}{Tho90}
\bibitem{Tho90}
\bibinfo{author}{B.~\surnamestart Thomsen\surnameend} (\bibinfo{year}{1990}):
  \emph{\bibinfo{title}{Calculi for Higher Order Communicating Systems}}.
\newblock \bibinfo{type}{Phd thesis}, \bibinfo{school}{Department of Computing,
  Imperial College}.

\bibitemdeclare{article}{Tho93}
\bibitem{Tho93}
\bibinfo{author}{B.~\surnamestart Thomsen\surnameend} (\bibinfo{year}{1993}):
  \emph{\bibinfo{title}{Plain {CHOCS}, a Second Generation Calculus for
  Higher-Order Processes}}.
\newblock {\sl \bibinfo{journal}{Acta Informatica}} \bibinfo{volume}{30(1)},
  pp. \bibinfo{pages}{1--59}, \doi{10.1007/BF01200262}.

\bibitemdeclare{article}{Xu12}
\bibitem{Xu12}
\bibinfo{author}{Xian \surnamestart Xu\surnameend} (\bibinfo{year}{2012}):
  \emph{\bibinfo{title}{Distinguishing and Relating Higher-order and
  First-order Processes by Expressiveness}}.
\newblock {\sl \bibinfo{journal}{Acta Informatica}} \bibinfo{volume}{49(7-8)},
  pp. \bibinfo{pages}{445--484}, \doi{10.1007/s00236-012-0168-9}.

\bibitemdeclare{inproceedings}{Xu13}
\bibitem{Xu13}
\bibinfo{author}{Xian \surnamestart Xu\surnameend} (\bibinfo{year}{2013}):
  \emph{\bibinfo{title}{On Context Bisimulation for Parameterized Higher-order
  Processes}}.
\newblock In: {\sl \bibinfo{booktitle}{Proceedings of the 6th Interaction and
  Concurrency Experience (ICE 2013)}}, {\sl \bibinfo{series}{EPTCS}}
  \bibinfo{volume}{131}, pp. \bibinfo{pages}{37--51},
  \doi{10.4204/EPTCS.131.5}.

\bibitemdeclare{misc}{Xu16app}
\bibitem{Xu16app}
\bibinfo{author}{Xian \surnamestart Xu\surnameend} (\bibinfo{year}{2016}):
  \emph{\bibinfo{title}{Higher-order Processes with Parameterization over Names
  and Processes (with appendices)}}.
\newblock
  \urlprefix\url{http://basics.sjtu.edu.cn/~xuxian/express2016withappendices.pdf}.

\bibitemdeclare{inproceedings}{XYL15}
\bibitem{XYL15}
\bibinfo{author}{Xian \surnamestart Xu\surnameend}, \bibinfo{author}{Qiang
  \surnamestart Yin\surnameend} \& \bibinfo{author}{Huan \surnamestart
  Long\surnameend} (\bibinfo{year}{2015}): \emph{\bibinfo{title}{On the
  Computation Power of Name Parameterization in Higher-order Processes}}.
\newblock In: {\sl \bibinfo{booktitle}{Proceedings of 8th Interaction and
  Concurrency Experience (ICE 2015)}}, {\sl \bibinfo{series}{EPTCS}}
  \bibinfo{volume}{189}, pp. \bibinfo{pages}{114--127},
  \doi{10.4204/EPTCS.189.10}.

\end{thebibliography}
